\documentclass[final,3p,12pt,times]{elsarticle}

\usepackage{amssymb}
\usepackage{amsthm}
\usepackage{amsmath}
\usepackage{multirow}
\usepackage{subfigure}
\usepackage{algpseudocode}
\usepackage{algorithm}
\usepackage{multicol}
\usepackage{setspace}
\usepackage{comment}
\usepackage[running,mathlines]{lineno}

 \biboptions{sort}

\newtheorem{theorem}{Theorem}
\newtheorem{lemma}{Lemma}
\newtheorem{corollary}{Corollary}
\newtheorem{remark}{Remark}
\date{xx}
\journal{\copyright manuscript}

\newcommand{\tr}{\operatorname{tr}}
\newcommand{\I}{\mathcal{I}}
\begin{document}
\begin{frontmatter}

\title{Essential formulae for restricted maximum likelihood and its derivatives associated with the linear mixed models\tnoteref{fund}}

\author{Shengxin Zhu \fnref{xjtlu}}
\ead{Shengxin.Zhu@xjtlu.edu.cn}
\author{Andrew J Wathen \fnref{ox}}


\date{}
\address[xjtlu]{Department of Mathematics, Xi'an Jiaotong-Liverpool University}

\address[ox]{Numerical Analysis Group, Mathematical Institute, University of Oxford}
\tnotetext[fund]{ The research
was partially sponsored by the Engineering and
Physical Sciences Research Council (EPSRC, industry mathematics
knowledge transfer project, IM1000852) and is supported by National Natural Science of China (NSFC)(No.11501044), Jiangsu Science \& Technology general program(BK20171237), and partially supported by (NSFC No. 11571002, 11571047, 1161049, 11671051, 61672003) }

\begin{abstract}
The restricted maximum likelihood method enhances popularity of maximum likelihood methods for variance component analysis on large scale unbalanced data.
As the high throughput biological data sets and the emerged science on uncertainty quantification, such a method receives increasing attention.
Estimating the unknown variance parameters with restricted maximum likelihood method usually requires an nonlinear iterative method. Therefore proper formulae for the log-likelihood function and its derivatives play an essential role in practical algorithm design. It is our aim to provide a mathematical introduction to this method, and supply a self-contained derivation on some available formulae used in practical algorithms. Some new proof are supplied.
\end{abstract}

\end{frontmatter}
 \newcommand{\var}{\operatorname{var}}

\section{Introduction}
Recent advance in genome-wide association study involves large scale linear mixed models \cite{lip11,lis12,Yu06,zhang10,zhou12}. 
Quantifying random effects in term of (co-)variance parameters in the linear mixed model is receiving increasing attention\cite{T08}. Common
random effects are blocks in experiments or observational studies that
are replicated across space or time \cite{Fisher1935,QK02}. Other random
effects like variation among individuals, genotypes and species also appear
frequently. In fact, geneticists and evolutionary biologists have long began
to notice the importance of quantifying magnitude of variation among
genotypes and spices due to environmental factors \cite{Fisher1930,Hend59,MH08},
Ecologist recently are interested in the importance of random
variation in space and time, or among individual in the study of
population dynamics \cite{Bolker2008,PS03}. Similar problems also arises in estimating parameter
in high dimensional Gaussian distribution \cite{Daniel12}, functional
data analysis \cite{Bart14}, model selection analysis \cite{MSW13,YMO14} and many other applications \cite{RK88}.

Quantifying such random effects and making a statistical inference requires estimates of the co-variance parameters in the underlying model. The estimates are
usually obtained by maximizing a log-likelihood function which
often involves nonlinearly log-determinant terms. The first derivative of the log-likelihood is often referred to as a \emph{score function}. To maximize the log-likelihood, one requires to find the zeros of the score functions according to the conceptually simple Newton Method.  However, the negative Jacobian matrix of score function, which is often referred to as the \emph{observed information matrix} is very complicated (see \cite[p.825, eq.8]{zhou12},\cite[p.26, eq 11]{meyer96} and $\I_{O}(\theta_i,\theta_j)$ in Table 1).
A remedy is the Fisher's scoring algorithm which uses the \emph{Fisher information matrix}
 in stead of the observed matrix \cite{jenn76}\cite{Longford1987}. The Fisher information matrix is simper than the Jacobian matrix but still involves a trace term of four matrix-matrix product(see $\I(\theta_i,\theta_j)$ in Table 1). Such a trace term is computationally prohibitive for large data sets.
 For variance matrices which linearly depend on the underlying variance parameters, \cite{john95} introduced the \emph{average information} matrix
 (see Table 1 $ \I_A(\theta_i,\theta_j)$) which only involves a quadratic form. It can be efficiently computed by matrix vector multiplications. Such an average information matrix serves one of the basis of the averaged information restricted maximum likelihood algorithm \cite{GTC95,M97}. For general variance matrix, this \emph{average information matrix} is the main part of the exact average of the observed and Fisher information matrices, the negligible part which involves a lot of computation is a random zero matrix \cite{Z16,ZGL16}. After some matrix transform, the "average information" can be computed by solving a sparse linear system with multiple right hand sides. Together with an efficient sparse factorization algorithm, it enables the derivatives methods work for high throughput biological data sets \cite{WZW13}\cite{Z17}. 

Derivative free \cite{gras87} methods have been studied. They require less computational time per iteration, but they converge slow and require more iterates, especially for large scale problems \cite{misz93}. Comparisons in \cite{misz94b} shows that the derivative approach requires less time for most cases. That is why recent large scale genome wide association applications  \cite{lip11,lis12,zhang10,zhou12} and robust software development prefer the derivative approach. In this paper, we focus on essential formulae used in an derivative method.

This papers aims to provide a self-contained derivation on these essential formulae used in practical algorithms. Most of the formulae are available in publications in statistics \cite{H77}, animal breeding \cite{gras87}\cite{Mey89} and quantitative genetics\cite{meyer96}\cite{M97}. These publications are written in a statistician perspective and omitted some necessary brief proof which prevent general algorithm designers to follow. Even some recent algorithms still employ an out-of-date formula \cite[p.825,eq.8]{zhou12}.
One of our aims is to fill such a gap. Therefore, we focus on brief mathematical derivation on this formulae. For more statistical introduction to the linear mixed model and applications, the reader is directed to the review articles \cite{S88}, its application in quantitative genetics \cite{T08} and the classical book by Searle et al \cite[Chapter 6.6]{Searle06}. Most of the proof supplied here are new and derived independently. We have tried our best to attribute these results to other existing results, if any.
Since there are voluminous publications on linear mixed models, we apology if there are some relevant work we haven't noticed yet. The derivation on the derivatives are simper and and brief than previous derivations in \cite{W94}, this may shed light on general variance parameter estimation scheme for the Gaussian process \cite{W94}\cite{wood11}. The derivations shows that evaluating the restricted log-likelihood function and its approximate second derivatives is closely related to efficient sparse factorization techniques. These formulae together with efficient sparse matrix techniques enables derivative maximum likelihood methods to work on large scale biological data set \cite{WZW13}\cite{Z17}.

The reminder of the paper is organised as follows. In Section 2, we shall introduce the variance parameter estimation problem associated to the linear mixed model. Followed by an existence theorem on choosing an \emph{error contrast} transform to derive the restricted maximum likelihood and three equivalent formulae for the restricted log-likelihood in Section 3. In section 4, we provide derivations on the first derivatives and seconde derivatives. Section 5 discuss computing issues. The paper is concluded with some discussion and remarks.



\begin{table*}[t!]
\centering
 \caption{Elements of the observed information($\I_A$), Fisher information($\I$) and averaged information splitting matrix($\I_A)$}
 \label{tab:splitting}
\begin{tabular}{ll}
\hline
$\I_O(\theta_i, \theta_j) $ &  $ \frac{1}{2}\left\{\tr(P\dot{V}_{ij})-
\tr(P\dot{V}_iP \dot{V}_j) +2y^TP\dot{V}_iP\dot{V}_j Py -y^T
 P\ddot{V}_{ij}Py\right\} $\\
$\I(\theta_i, \theta_j)$ & $  \frac{1}{2} \tr(P\dot{V}_iP\dot{V}_j) $\\
$\I_A(\theta_i, \theta_j)$ & $\frac{1}{2} y^TP\dot{V}_i P \dot{V}_j P y$ \\
 & $P= V^{-1} - V^{-1}X(X^TV^{-1}X)^{-1}X^T V $ \\
 & $\dot{V}_i = \frac{\partial V}{\partial \theta_j}$, $\ddot{V}_{ij}=\frac{\partial^2 V}{ \partial \theta_i \partial \theta_j}$ \\
  \hline
\end{tabular}
\end{table*}

\section{Preliminary}

The basic model we considered is the widely used Linear Mixed Model(LMM),
\begin{equation}
y=X\tau+Zu+e. \label{eq:LMM}
\end{equation}
In the model, $y\in \mathbb{R}^{    n\times 1}$ is a vector of observable
measurements
 $\tau\in \mathbb{R}^{p\times 1}$ is a vector of fixed
effects, $X \in \mathbb{R}^{n\times p}$ is a \textit{design
matrix} which corresponds to the fixed effects, $u \in
\mathbb{R}^{b\times 1}$ is a vector of random effects, $Z \in
\mathbb{R}^{n\times b}$ is a design matrix which corresponds
to combination of random effects.
$e\in \mathbb{R}^{n\times 1}$ is the vector of residual errors.
The linear mixed model is an extension to the linear model
\begin{equation}
y=X\tau +e.    \label{eq:LM}
\end{equation}
LMM allows additional random components, $u$, as
correlated error terms, the linear mixed model is also referred to as \textit{linear mixed-effects models}. The term(s) $u$ can be
added level by level, therefore it is also referred to as \textit{hierarchal models}. It brings a wider range of
\textit{variance structures and models} than the linear model in
\eqref{eq:LM} does. For instance, in most cases, we suppose that
the random effects, $u$, and the residual errors, $e$, are
multivariate normal distributions such that $E(u)=0$, $E(e)=0$,
$u\sim N(0, \sigma^2 G)$, $e\sim N(0, \sigma^2 R)$ and
\begin{equation}
\text{var}\left[\begin{array}{c}
u\\
e
\end{array}\right]=\sigma^{2}\left[\begin{array}{cc}
G(\gamma) & 0\\
0 & R(\phi)
\end{array}\right],
\end{equation}
where $G\in \mathbb{R}^{b\times b }$, $R \in \mathbb{R}^{n\times
n}$. We shall denote $\kappa=(\gamma; \phi)^T.$
Under these
assumptions, we have
\begin{align}
y \vert  u &\sim N(X\tau+Zu, \sigma^2 R), \\
y &\sim N(X\tau, \sigma^2(R+ZGZ^T)):= N(X\tau, V(\theta)),
\end{align}
where $\theta=(\sigma^2;\kappa)^T$.
When the co-variance matrices $G$ and $R$ are known, one can obtain the \emph{Best Linear Unbiased Estimators} (BLUEs), $\hat{\tau}$, for the fixed effects and the \emph{Best Linear Unbiased Prediction} (BLUP), $\tilde{u}$, for the random effects according to the maximum likelihood method, or the Gauss-Markov-Aitiken least square \cite[\S 4.2]{Rao}. $\hat{\tau}$ and $\tilde{u}$ satisfy the following mixed model equation \cite{Hend59}
\begin{equation}
\begin{pmatrix}
X^TR^{-1}X  & X^T R^{-1}Z \\
Z^TR^{-1}X  & Z^TR^{-1}Z+G^{-1}
\end{pmatrix}\begin{pmatrix}
\hat{\tau}  \\ \tilde{u}
\end{pmatrix}
=\begin{pmatrix}
X^TR^{-1}y \\ Z^T R^{-1}y
\end{pmatrix}.
\label{eq:mme}
\end{equation}
For such a forward problem, confidence or uncertainty of the estimations of the fixed and random effects can be quantified in term of co-variance of the estimators and the predictors
\begin{equation}
\mathrm{var} \begin{pmatrix}
\hat{\tau}-\tau \\
\tilde{u} -u
\end{pmatrix}
 =\sigma^2 C^{-1},
\label{eq:pre}
\end{equation}
where $C$ is the coefficient matrix in the mixed model equation \eqref{eq:mme}.

In many other more realistic and interesting cases. The variance parameter $\theta$ is unknown and to be estimated. This paper focuses on these cases.
One of the commonly used methods to estimate variance parameters is
the maximum likelihood principle. In this approach, one starts with the distribution of the random vector $y$.
The variance of $y$ in the linear mixed model \eqref{eq:LMM} is
\begin{equation}
V = \operatorname{var}(y) 
    =\sigma^{2}(R+ZGZ^{T}):=\sigma^2H(\kappa)=V(\theta),
\end{equation}
and the likelihood function of $y$ is
\begin{equation}
L(\tau,\theta)=
\prod_{i=1}^n (2\pi)^{-\frac{n}{2}} |V(\theta)|^{-\frac{1}{2}}
\exp\left\{-\frac{1}{2}(y-X\tau)^TV(\theta)^{-1}(y-X\tau) \right\}.
\end{equation}
Since the logarithmic transformation is monotonic, it is equivalent to maximize
$\log L(\tau, \theta)$ instead of $L(\tau,\sigma^2)$.  The log-likelihood function is \cite[Chapter 6.2, eq.(13)]{Searle06}
\begin{equation}
\log L(\tau,\theta)=-\frac{1}{2}\left\{n\ln(2\pi)+\ln|V(\theta)| +(y-X\tau)^TV(\theta)(y-X\tau) \right\}.
\end{equation}
A maximum likelihood estimates for the variance parameter $\theta$ is
$$
\hat{\theta}=\arg_{\theta}\max \log L(\tau,\theta).
$$
The maximum likelihood estimate, $\hat{\sigma}^2$, for the variance parameter is asymptotically approaching to the true value,  $\sigma^2$, however, the
bias is relative large for finite observations with relative many effective fixed effects. Precisely
$$
\mathrm{Bias}(\hat{\sigma}^2,\sigma^2)=\frac{\nu}{n} \sigma^2,
$$
where $\nu=\mathrm{rank}(X)$. A remedy to remove or at least reduce such a bias is the Restricted Maximum Likelihood (REML) \cite{PT71}, which is also referred to as the \textit{marginal maximum likelihood} method or \textit{REsidual
Maximum Likelihood} method.
In original derivation of REML, the observation $y$ is transformed into two statically independent parts, $Sy$ and $Qy$ such that $\mathrm{cov}(Sy,Qy)=0$, where $S=I-X(X^TX)^{-1}X$ is a projection matrix with rank $\nu=\mathrm{rank}(X)$ such that $E(Sy)=0$, and $Q=X^TV^{-1}$ is a weighted project matrix with rank $\nu$. The likelihood for $Sy$ does not involves the information related to the fixed effects. Such a transform $Sy$ such that $E(Sy)=0$ is referred to as the \emph{error contrast} \cite{Harville}. A simpler error contrast to derive REML was suggested in \cite{Verbyla1990}:
\newcommand{\rank}{\operatorname{rank}}
for any $X\in
\mathbb{R}^{n\times p}$, choose a linear transformation
$L=[L_1, L_2]$, such that $L_1^TX=I_p$ and $L_2^TX=0$ (
Theorem \ref{thm:L} provides an independent proof on how to choose such an error contrast transformation).
\begin{equation}
L^Ty=\begin{pmatrix}
L_1^Ty \\ L_2^Ty
\end{pmatrix}
\sim
N\left( \begin{pmatrix}
\tau \\0
\end{pmatrix},
\begin{pmatrix}
L_1^TVL_1 & L_1^TVL_2 \\
L_2^TVL_1 & L_2^TVL_2
\end{pmatrix}
 \right).
 \label{eq:y1y2}
\end{equation}
The fixed effects are determined by maximizing the log-likelihood function of $L_1^Ty$.
The marginal distribution of $L_2^Ty$
\cite[p40, Thm
2.44]{all04}
$$y_2=L_2^Ty\sim N(0,L_2^TVL_2) $$
is used to derive the restricted likelihood:
\begin{equation}
\ell_R=-\frac{1}{2}\{ (n-\nu )\log(2\pi) +\log \vert L_2^T V(\theta) L_2 \vert + y^TL_2(L_2^TV(\theta)L_2)^{-1}L_2^Ty \}. \label{eq:l2}
\end{equation}
The REML estimate for the variance parameter is
$$
\hat{\theta}^{\mathrm{REML}}=\arg_{\theta}\max\ell_R(\theta).
$$
Such an estimate removes redundant
freedoms which are used in estimating the fixed effects and is often unbiase. Because of such an unbiased estimation, the REML method enhances the popularity of restricted maximum likelihood methods.

\begin{algorithm}[!t]
\caption{Newton-Raphson method to solve {$S(\theta)=0$.}}
\begin{algorithmic}[1]
\State {Give an initial guess of $\theta_0$}
\For{ $k=0, 1, 2, \cdots$ until convergence }
 \State{Solve $\I_o(\theta_k) \delta_k=S(\theta_k)$}
 \State{$\theta_{k+1}=\theta_k+\delta_k$}
\EndFor
\end{algorithmic}
\label{alg:NR}
\end{algorithm}

\section{Error contrast transform and closed formulae for the restricted log-likelihood}
We first provide a rigorous proof on the existence of the the error contrast transform and thus how to construct the error contrast transform. We shall first introduce the following lemma.
\newtheorem{thm}{Theorem}
\newtheorem{lem}{Lemma}
\subsection{On choice of the error contrast transform}
\begin{lemma}
Let $X\in \mathbb{R}^{n\times p}$ be full rank and $P_X=X(X^TX)^{-1}X^T$,
then there exists an orthogonal matrix $K=[K_1, K_2]$, such that
\begin{enumerate}
  \item $P_X=K_1K_1^T$;
  \item $I-P_X=K_2K_2^T$.
\end{enumerate}
\label{thm:a}
\end{lemma}
\begin{proof}
It is easy to verify that $P_X$ is an
symmetric \textit{projection/idempotent matrix}, i.e.
$$P_X^T=P_X,\quad P_X^2=P_X.$$ Since $P_X(I-P_X)=0$, the eigenvalues of
$P_X$ are 1 and 0. There exists an orthogonal matrix $K=(K_1,
K_2)$, $K\in \mathbb{R}^{n\times n}$, $K_1\in
\mathbb{R}^{n\times p}$, and $K_2\in \mathbb{R}^{n\times (n-p)}$
such that
\begin{equation}
P_X=(K_1, K_2) \begin{pmatrix}
I_p & 0 \\ 0 & 0
\end{pmatrix}
\begin{pmatrix}
K_1^T \\K_2^T
\end{pmatrix}=K_1K_1^T.
\end{equation}
One can show that there are exactly $p$ eigenvalues with 1.

Equivalently,
 \begin{equation}
P_X(K_1, K_2)=(K_1, K_2)\begin{pmatrix}
I_p & 0 \\ 0 & 0
\end{pmatrix}.
\end{equation}
It is clear that each column of $K_1$($K_2$) is an eigenvector of $P_X$ corresponding
to the eigenvalue 1(0). Further according to $P_XX=X$, each column of $X$
is an eigenvector corresponding to $1$. Since eigenvectors corresponding to different
eigenvalues are orthogonal, we have
\begin{equation}
K_2^TX=0.
\end{equation}
Further, one can verify that
\begin{align*} I&=(KK^T)(KK^T)=(K_1, K_2)\begin{pmatrix} K_1^T\\K_2^T\end{pmatrix}(K_1,K_2)\begin{pmatrix} K_1^T \\K_2^T \end{pmatrix}
 =(K_1,K_2)\begin{pmatrix} K_1^T K_1  & 0 \\
 0 & K_2^TK_2 \end{pmatrix} \begin{pmatrix}  K_1^T
 \\K_2^T
 \end{pmatrix}\\
 &=K_1(K_1^TK_1)K_1^T+K_2(K_2^TK_2)K_2^T=K_1K_1^T+K_2K_2^T.
\end{align*}
We have
\begin{equation}
 I-P_X=K_2K_2^T.   \label{eq:K2K2}
 \end{equation}
  \flushright \qed
\end{proof}
\begin{theorem}
Let $X\in \mathbb{R}^{n\times p}$ and $\rank{X}=p$, $p<n$. Then
there exist nonsingular matrices $L=[L_1, L_2]$, such that
$L_1^TX=I_{p\times p}$, $L_2^TX=0_{(n-p)\times p}$.
\label{thm:L}
\end{theorem}
\begin{proof}
Let $B\in \mathbb{R}^{(n-p)\times (n-p)}$ is any nonsingular
matrix and $K_2K_2^T=I-P_X$ be defined in \eqref{eq:K2K2}. Then
$BK_2^TX=0$ ($K_2B^T\in \ker{X^T}$) and $\rank=K_2B^T=n-p$.
Therefore the columns of $\{X, K_2B^T\}$ form a set of basis of
$\mathbb{R}^{n\times n}$. Denote $L^T=[X, K_2B^T]^{-1}$, then
use the identy $L^T[X,K_2B^T]=I$ , we have
\begin{equation}
\begin{pmatrix} L_1^TX  & L_1^TK_2B^T \\L_2^TX & L_2^TK_2B^T
\end{pmatrix}
=
\begin{pmatrix}
I_{p\times p} &0\\0&I_{(n-p)\times(n-p)}
\end{pmatrix}
\end{equation}
\flushright \qed
\end{proof}

\begin{remark}
Theorem \ref{thm:L} indicates that there are some freedom to choose an error contrast transform.
In the original derivation of the REML \cite{PT71}, the authors use the projection matrix $S=I-X(X^TX)^{-1}X^T=K_2K_2^T$ as an error contrast transform from $\mathbb{R}^n$ to $\mathbb{R}^n$, this results a singular variance matrix, $SVS$, for $Sy$. One has to work with the general inverse of $SVS$ to derive the log-likelihood function. The simplest choice of an error contrast transform from $\mathbb{R}^n$ to $\mathbb{R}^{n-p}$, is the transpose of the last $n-\nu$ columns of the inverse of $[X, K_2]$.
\end{remark}

\subsection{Closed formulae of the restricted log-likelihood function}

The restricted log-likelihood given in \eqref{eq:l2} involves an intermediate matrix $L_2$. We shall prove this formula is equivalent to the close formula given in \cite{H77}
\begin{equation*}
\ell_{R} = -\frac{1}{2}\left\{ \mathrm{const}+\log |H|
+\log|X^TH^{-1}X| +y^TPy \right\} ,
\end{equation*}
and the following formula used in a derivative-free approach \cite{gras87}\cite{Mey89}\cite{mey05} \cite{M97}(see also \cite{meyer96}\cite{mey05})
\begin{equation*}
\ell_R=-\frac{1}{2} \left\{\mathrm{const} +\log\lvert R \rvert +\log \lvert G \rvert + \log\lvert C\lvert +y^TPy \right\},
\end{equation*}
where $P=V^{-1} -V^{-1}X (X^TV^{-1}X)^{-1}X^TV^{-1}$ and $C$ is the coefficient matrix of the mixed model equation.
To prove the equivalence, we first introduce the following lemmas.

\newcommand{\R}{\mathbb{R}}
\begin{lemma}
Let $X\in \mathbb{R}^{n\times p}$ be full rank and $P_X=X(XX)^{-1}X^T$. For any full
rank matrix $L_2\in \R^{n\times(n-p)}$, and $L_2^TX=0$, we have
\begin{equation}
I-P_X=L_2(L_2^TL_2)^{-1}L_2^T,
\end{equation}
\label{thm:p2}
\end{lemma}
\begin{proof}
Let $B=[X,L_2]$. Since the columns of B is linear independent,
therefore we have the identity
$I=BB^{-1}B^{-T}B^T=B(B^TB)^{-1}B^T$.
\begin{equation}
I=(X, L_2)\begin{pmatrix} X^TX & X^TL_2 \\
L_2^TX & L_2^TL_2
\end{pmatrix}^{-1}\begin{pmatrix} X^T \\ L_2^T
\end{pmatrix}.
\end{equation}
Use $L_2^TX=0$, we have $P_X+L_2(L_2^TL_2)^{-1}L_2^T=I$.  \flushright \qed
\end{proof}

\begin{lemma} Let $V\in \R^{n \times n}$ be a symmetric positive definite
matrix. $X\in \R^{n\times p}$,  $P_X^V=V^{-1}-V^{-1}X(X^TV^{-1}X)^{-1}X^TV^{-1}$ and $L=[L_1, L_2]\in \R^{n\times
n}L_2$ such that $L_1^TX=I_p$, $L_2^TX=0$, then
\begin{equation}
P_X^V=L_2(L_2^TVL_2)^{-1}L_2^T,
\label{eq:L2HL2}
\end{equation}
and
 \begin{equation}
 (X^TV^{-1}X)^{-1}=L_1^TVL_1-L_1^{T}VL_2(L_2^TVL_2)^{-1}L_2^TVL_1.
 \label{eq:XHX}
 \end{equation}
 \label{lem:d}
\end{lemma}
\begin{proof}
The equation \eqref{eq:L2HL2} is due to \cite{K66}(see also \cite[Appendix M.4f]{Searle06}). Here it is a directly consequence of Lemma \ref{thm:p2}.
Since $V$ is symmetric positive definite, then there exists a
symmetric positive definite $V^{1/2}$.
Let $\hat{X}=V^{-1/2}X$, then for  $\hat{X}\in
\R^{n\times p}$,  $L_2^TV^{1/2} \hat{X}=0$.  According to Theorem
\ref{thm:p2}, we have
\begin{equation}
I-P_{\hat{X}}=V^{1/2}L_2(L_2^TVL_2)^{-1}L_2^TV^{1/2}.
\end{equation}
Multiply $V^{-1/2}$ on left and right on both side of the
equation, we obtain
\begin{equation}
V^{-1}-V^{-1}X(X^TV^{-1}X)^{-1}X^TV^{-1} = L_2(L_2^TV^{-1}L_2)^{-1}L_2.
\end{equation}

Using the equation \eqref{eq:L2HL2} on the right hand side of
\eqref{eq:XHX},
we have
\begin{align*}
   &L_1^TVL_1- L_1^TV\underbrace{L_2(L_2^TVL_2)^{-1}L_2^T}_{=P_X^V}VL_1
  =L_1^TVL_1-L_1^T\underbrace{(V-X(X^TV^{-1}X)^{-1}X^T)}_{=HP_X^VV}L_1 \\
  = &
  \underbrace{L_1^TX}_{=I_p}(X^TV^{-1}X)^{-1}\underbrace{X^TL_1}_{=I_p}=(X^TV^{-1}X)^{-1}.
\end{align*}
\flushright \qed
\end{proof}

\begin{theorem}
The residual log-likelihood for the linear mixed model can be written as follows:
\begin{align}
\ell_{R} & = -\frac{1}{2}\{ (n-\nu )\log(2\pi) +\log \vert L_2^T V L_2 \vert + y^TL_2(L_2^TVL_2)^{-1}L_2^Ty \}, \label{eq:lra}\\
         & = -\frac{1}{2}\left\{ (n-\nu)\log ( 2\pi)+\log |V|+\log|X^TV^{-1}X| +y^TPy \right\}. \label{eq:lrb}\\
         & = -\frac{1}{2} \left\{(n-\nu )\log(2\pi \sigma^2) + \log\lvert R \rvert +\log \lvert G \rvert + \log\lvert C\lvert +y^TPy \right\}. \label{eq:lrc}
\end{align}
where $ V=V(\theta)=\sigma^2(R+ZGZ^T)$ and
\begin{equation}
P=V^{-1}-V^{-1}X(X^TV^{-1}X)^{-1}X^TV^{-1}.\label{eq:PHX}
\end{equation}
\end{theorem}

\begin{proof} The formulae \eqref{eq:lra} and \eqref{eq:lrb} are standard and can be found in standard text book, see  \cite[Chapter 6.6.e]{Searle06}, The formula \eqref{eq:lrc} are often used in a derivative-free approach \cite{gras87}. Here we provide a unified proof.
According to Lemma ~\ref{lem:d}, we have $P=L_2(L_2^TVL_2)^{-1}L_2$ and
$$
 (X^TV^{-1}X)^{-1}=L_1^TVL_1-L_1^{T}VL_2(L_2^TVL_2)^{-1}L_2^TVL_1.
$$
Then we use the identity
\begin{align*}
 & \left\vert \begin{pmatrix}
I_p & -(L_1^TVL_2)(L_2^TVL_2)^{-1} \\
0 &I_{n-p}
\end{pmatrix}
\begin{pmatrix}
L_1^TVL_1 & L_1^TVL_2 \\
L_2^TVL_1 & L_2^TVL_2
\end{pmatrix} \right\vert
 =\left\vert
\begin{pmatrix}
(X^TV^{-1}X)^{-1} & 0\\
L_2^TVL_1 & L_2^TVL_2
\end{pmatrix}\right\vert=\vert L^TVL \vert,
\end{align*}
We have $|L^TVL|=|L^T||V||L|=|(X^TV^{-1}X)^{-1}||L_2^TVL_2|$. This indicates that
\begin{equation}
\underbrace{\log|L^T| +\log|L|}_{0}+\log|V| + \log|X^TV^{-1}X| =\log|L_2^TVL_2^T|.
\end{equation}
which proves the equivalence between the first two formulae.

Notice that $V=\sigma^2H$, then
\begin{equation}
\log\lvert V \rvert +\log\lvert X^TV^{-1}X \rvert =(n-\nu) \log\sigma^2 +\log \lvert X^TH^{-1}X\rvert + \log \lvert H \rvert. \label{eq:la}
\end{equation}
Apply the Woodbury matrix identity \cite{W89} \cite[Fact 2.16.21]{B09}, we have
\begin{equation}
H^{-1}=(R+ZGZ^T)^{-1}=R^{-1} -R^{-1}Z(G^{-1}+ZR^{-1}Z^T)^{-1}Z^TR^{-1}.  \label{eq:H-1}
\end{equation}
Therefore we have
$$
X^{T}H^{-1}X=X^TR^{-1}X -X^TR^{-1}Z(G^{-1}+ZR^{-1}Z^T)^{-1}Z^TR^{-1}X.
$$
This is nothing but the Schur complement for the block elimination of the matrix $C$
\begin{equation*}
\begin{pmatrix}
X^TR^{-1}X & X^TR^{-1} Z \\
X^TR^{-1}Z & Z^TR^{-1}+G^{-1}
\end{pmatrix}
\begin{pmatrix}
I  &  0 \\
F & I
\end{pmatrix}
=
\begin{pmatrix}
XH^{-1}X  & *  \\
 0   & Z^TR^{-1}Z +G^{-1}
\end{pmatrix},
\end{equation*}
where $F=-(Z^TR^{-1}Z+G^{-1})^{-1} Z^TR^{-1}X.$
Therefore we have
\begin{equation}
\log \lvert C \rvert =\log \lvert XH^{-1}X \rvert +\log \lvert Z^TR^{-1}Z +G^{-1} \rvert.  \label{eq:lb}
\end{equation}
Now consider the block elimination of the following matrix
\[
\begin{pmatrix}
R & Z \\
-Z^T & G^{-1}
\end{pmatrix}
\begin{pmatrix}
I & 0 \\
GZ^T & I
\end{pmatrix}
=\begin{pmatrix}
R+ZGZ^T & R \\
0 & G^{-1}
\end{pmatrix}=\begin{pmatrix}
H & G  \\
0 & G^{-1}
\end{pmatrix}.
\]
Similarly, we have
\[
\begin{pmatrix}
I & 0 \\
Z^TR^{-1} & I
\end{pmatrix}
\begin{pmatrix}
R & Z \\
-Z^T & G^{-1}
\end{pmatrix}
=\begin{pmatrix}
R & Z\\
0 & G^{-1} + Z^TR^{-1}Z
\end{pmatrix}.
\]
Therefore we have
\begin{equation}
\lvert R\rvert \lvert G^{-1}+Z^TR^{-1}Z \rvert =\lvert H\rvert \lvert G^{-1} \rvert. \label{eq:LR3}
\end{equation}
Therefore we have
\begin{equation}
\log\lvert C \rvert +\log\lvert R\rvert +\log \lvert G \rvert  =\log \lvert H \rvert  +\log \lvert X^TH^{-1}X \rvert. \label{eq:lc}
\end{equation}
Combine \eqref{eq:la} and \eqref{eq:lc}, we conclude the equivalence between \eqref{eq:lrb} and \eqref{eq:lrc}.
\flushright \qed
\end{proof}

\section{Scores and its derivatives for REML}

The first derivatives of a log-likelihood function is referred to as a \emph{score function}. The negative Jacobian matrix of the score function, or the negative Hessian matrix of the log-likelihood, is called the \emph{observed information}. The \emph{Fisher information matrix} is the expect value of the the observed information matrix. These derivatives and related terms of a log-likelihood function play an important role in an derivative approach to estimate the variance parameters. Formulae here are based on the pioneering   work by \cite{H77}, \cite{jenn76},\cite{Mey89},\cite{meyer96} \cite{M97}. The classical textbook \cite[Chapter 6.6]{Searle06} derives some of the formulae in the case when the variance matrices satisfies $V=\sum_{i=1}^b \sigma^2 Z_i Z_i^T$. We shall follow the frame work given in \cite[Chapter 6.6]{Searle06}.    The idea of \textit{averaged information splitting theorem} generalized the \emph{average information} for variance matrix with linear variance structure proposed by \cite{john95}. The authors of \cite{W94} use formulae \eqref{eq:lrb} and divide $\ell_R$ into three part and works on them separately. The derivation provided here is brief and simpler than those in \cite{W94}.

\subsection{The score functions for residual log-likelihood}

\begin{theorem}[\cite{H77}]
Let $X\in \mathbb{R}^{n\times p}$ be full rank in the linear mixed model \eqref{eq:LMM} and the restricted log-likelihood function be given as $\ell_R(\theta)$ in \eqref{eq:l2}.  The scores of the residual
log-likelihood $\ell_R$ are given by
\begin{equation}
s(\theta_i)= -\frac{1}{2}\{ \tr(P \dot{V}_i)-y^TP\dot{V}_iPy \}
\end{equation}
where $\dot{V}_i=\frac{\partial V}{\partial \theta_i}$ and $P=V^{-1} -V^{-1}X (X^TV^{-1}X)^{-1}X^TV^{-1}$.
\end{theorem}
\begin{proof} This formula and the following derivation are based on \cite[p.252, eq.(91) to eq. (93)]{Searle06} for variance matrix $V$ with linearly dependence on $\theta$.
 First, according to Lemma~\ref{lem:d}, We know that $P=L_2(L_2^TVL_2)^{-1}L_2^T)$.
\begin{equation}
s(\theta_i)=\frac{\partial\log |L_2^THL_2|}{\partial \theta_i}+
\frac{\partial (y^TVy)}{\partial \theta_i}.
\label{eq:sk}
\end{equation}
 Using the fact on matrix derivatives of log determinant\cite[p.305, eq.8.6]{Har97}
 $$\frac{\partial \log|A|}{\partial
\theta}=\tr(A^{-1}\frac{\partial A}{\partial \theta_i})$$
and the property of the trace
operation $\tr(AB)=\tr(BA)$ \cite[p.50, eq. (2.3)]{Har97}
\begin{align}
&\frac{\partial \log(\lvert L_2^TVL_2\rvert)}{\partial \kappa_i} =
\tr\left((L_2VL_2)^{-1} \frac{\partial (L_2^TVL_2)}{\partial
\kappa_i}\right)
=\tr\left( \underbrace{L_2(L_2^TVL_2)^{-1} L_2^T}_{=P}
\dot{V}_i\right)=\tr\left(P\dot{V}_i\right).
\end{align}
 For the second term in \eqref{eq:sk}, we apply the result on
on matrix derivatives of the inverse of a matrix
\cite[p.307,eq.8.15]{Har97}
\begin{equation*}
\frac{\partial A^{-1}}{ \partial \theta_i}
=-A^{-1}\frac{\partial A}{\partial \theta_i} A^{-1}.
\end{equation*}
We have \cite[p. 252, eq.(91)]{Searle06}
\begin{align}
\frac{\partial P}{\partial \theta_i} &=\frac{\partial (L_2(L_2^TVL_2)^{-1}L_2^T)}{\partial
\theta_i}=L_2\frac{\partial (L_2^TVL_2)^{-1}}{\partial
\theta_i}L_2^T  \nonumber \\
& =-L_2(L_2^TVL_2)^{-1}\frac{\partial (L_2^TVL_2)}{\partial
\theta_i} (L_2^TVL_2)^{-1}L_2^T \nonumber \\
& =-\underbrace{L_2(L_2^TVL_2)^{-1}L_2^T}_{=P}\dot{V}_i
\underbrace{L_2(L_2^TVL_2)^{-1}L_2^T}_{=P}
 =-P\dot{V}_iP
=\dot{P}_i \label{eq:dpxh}.
\end{align}
\flushright \qed
\end{proof}

\subsection{ Observed information matrix: negative Jacobian of the score}
The negative of the Hessian matrix of a log-likelihood function, or the negative Jacobian of the score function, is often refereed to as the
\textit{observed information matrix},
\begin{equation}
\I_o=  -\left(  \frac{\partial^2 \ell_R}{ \partial \theta_i \partial \theta_j} \right).
\end{equation}
In term of the observed information matrix, line 3 in Algorithm \ref{alg:NR} reads
as $\I_{o} \delta_k=S(\theta_k).$
\begin{theorem}[\cite{H77}]
Elements of the observed information matrix for the residual
log-likelihood \eqref{eq:l2} are given by
\begin{equation}
\I_o(\theta_i,\theta_j) = \frac{1}{2}\left\{\tr(P\dot{V}_{ij})-
\tr(P\dot{V}_iP \dot{V}_j) +2y^TP\dot{V}_iP\dot{V}_j Py -y^T
 P\ddot{V}_{ij}Py\right\}.
\label{eq:IKK}
\end{equation}
where $\dot{V}_i=\frac{\partial V}{\partial \kappa_i}$,
$\ddot{V}_{ij}=\frac{\partial^2 V}{\partial \kappa_i \partial \kappa_j}$.
 \end{theorem}
 \begin{proof}
 The first two terms in \eqref{eq:IKK} follows by applying the result in \eqref{eq:dpxh},
 \begin{align*}
 \frac{\partial \tr(P\dot{V}_i)}{\partial
 \theta_j}  =\tr(P\ddot{V}_{ij})+\tr(\dot{P}_j\dot{V}_i)
 =\tr(P\ddot{V}_{ij})-\tr(P\dot{V}_jP\dot{V}_i).
 \end{align*}
 The last two terms in \eqref{eq:IKK} follows because of
 the result in \eqref{eq:dpxh},
 we have
 \begin{equation}
 -\frac{\partial (P\dot{V}_iP)}{\partial
 \theta_j}=P\dot{V}_jP\dot{V}_iP-P\ddot{V}_{ij}P+P\dot{V}_iP\dot{V}_jP.
 \end{equation}
 Further note that $\dot{V}_i$, $\dot{H}_j$ and $P$ are symmetric.
The second term in \eqref{eq:IKK} follows because of
 $$y^TP\dot{V}_iP\dot{V}_jPy=y^TP\dot{V}_jP\dot{V}_iPy. $$

  \flushright \qed
 \end{proof}

The elements \eqref{eq:IKK} in the observed information matrix, the negative Jacobian matrix of the score, involve the trace product of four matrices. It is computationally prohibitive for large data
set. Therefore it is necessary to approximate the Jacobian matrix for efficiency.
\begin{corollary}
If $V$ be a variance-covariance matrix which depends linearly on the variance parameter, say, $V=\sum \theta_i Z_iZ_i^T$, then $\ddot{V}_{ij}=0$, in this case the element of the observed information matrix are given
\begin{equation}
\I_O(\theta_i, \theta_j) = y^TP\dot{V}_iP\dot{V}_j Py -\frac{1}{2} \tr(P\dot{V}_i P \dot{V}_j).
\end{equation}
\end{corollary}

\subsection{The Fisher information matrix}
The \textit{Fisher information matrix}, $\I$, is the
expect value of the observed information matrix, $ \I=E(\I_o).$
The Fisher information matrix has a simpler form than the
observed information matrix and provides essential information on the observations, and thus it
is a nature approximation to the observed information matrix.

\begin{theorem}[\cite{H77}]
Elements of the Fisher information matrix for the residual
log-likelihood function in \eqref{eq:l2} are given by
\begin{align}
\I(\theta_i, \theta_j)
&=E(\I_o(\theta_i,\theta_j))=\frac{1}{2}\tr(P\dot{V}_iP\dot{V}_j)
.\label{eq:FIKK}
\end{align}
\end{theorem}
\begin{proof} The formulas can be found in \cite{GTC95}. Here we supply an alternative proof.
First note that
\begin{align}
PX&=V^{-1}X-V^{-1}X(X^TV^{-1}X)^{-1}XV^{-1}X=0, \quad E(y)=X\tau. \nonumber \\
PV&=PE((y-X\tau)(y-X\tau)^T) =PE(yy^T-X\tau y^T -y(X\tau)^T +X\tau (X\tau)^T)=PE(y^Ty).
\label{eq:PEyy}
\end{align}
Then
\begin{equation}
E(y^TPy)=E(\tr(Pyy^T))=\tr(PE(yy^T)) =\tr(PV). \label{eq:Eypy}
\end{equation}
Notice that $PVP=P$. Apply the procedure in
\eqref{eq:Eypy}, we have
\begin{align}
E(y^TP\dot{V}_iP\dot{V}_jPy) &=\tr(P\dot{V}_iP\dot{V}_jPV)=\tr(PVP\dot{V}_iP\dot{V}_j) =\tr(P\dot{V}_iP\dot{V}_j), \label{eq:yphphpy}
\end{align}
This proves the result.
 \flushright \qed
\end{proof}

\begin{remark}
For a variance matrix $V$ such that $\ddot{V}_{ij}=0$, Meyer and Smith \cite{Mey89}notice that
\[
 \frac{\partial ^2 \ell_R}{ \partial \theta_i \partial \theta_j} = y^TP\dot{V}_i P \dot{V}_j Py.
\]
and
\[
\I = \I_o  - \frac{1}{2} \frac{ \partial^2 y^TPy}{\partial \theta_i \theta_j} = -\frac{1}{2}\left\{ \frac{\partial^2 \log \lvert R  \rvert }{\partial  \theta_i \partial \theta_j } + \frac{\partial^2 \log \lvert G  \rvert }{\partial  \theta_i \partial \theta_j } + \frac{\partial^2 \log \lvert C  \rvert }{\partial  \theta_i \partial \theta_j } \right \}.
\]
Such a formula is used in \cite{misz94b} and \cite{misz93}. \cite{W94} also splits the derivatives into $G$ terms, $R$ terms and correlated terms(C).
\end{remark}

Using the Fishing information matrix as an approximate to the negative Jacobian result in the widely used Fisher-scoring algorithm \cite{Longford1987}. Jennrich Sampson found the Fisher-scoring algorithm is more robust to poor starting value \cite{jenn76}.
\begin{algorithm}[!t]
\caption{Fisher scoring algorithm to estimate the variance
parameters}
\begin{algorithmic}[1]
\State {Give an initial guess of $\theta_0$} \For{ $k=0, 1, 2,
\cdots$ until convergence }
 \State{Solve $\I(\theta_k) \delta_k=S(\theta_k)$}
 \State{$\theta_{k+1}=\theta_k+\delta_k$}
\EndFor
\end{algorithmic}
\label{alg:FS}
\end{algorithm}

\subsection{Average information matrix for variance matrices linearly depending variance parameters}

For variance matrix $V$ such that $\ddot{V}_{ij}=0$,  Johnson and Tompson \cite{john95} noticed that the \emph{average} of the observed and the Fisher information enjoys an computational efficient formula, \cite{john95} \cite{meyer96} \cite{M97}
\begin{equation}
\I_A(\theta_i, \theta_j) = \frac{\I_O(\theta_i +\theta_j) +\I(\theta_i, \theta_j)}{2} =\frac{y^TP\dot{V}_iP\dot{V}_jPy}{2}. \label{eq:Ia}
\end{equation}
We shall call $\I_A$ as the \emph{average information matrix}.
This formula is used in the average information REML algorithm \cite{GTC95}, which serves one of the foundation of the ASReml software package \cite{ASReml}.
There are quite a few cases when $\ddot{V}_{ij}=0$, for example, we reformulated the linear mixed model as
\[
 y=X\tau + Z_1 u_i +Z_2 u_2 + \cdot + Z_r u_r +\epsilon,
\]
where $u_i \in R^{b_i}$ are random effects in the level $i$, and $\mathrm{cov}(u_i, uj)=0$. Then the variance matrix $G$ has the formula
\[
 G=\begin{pmatrix}
                \sigma^2 I_i &  &  &  \\
                 & \sigma^2 I_2 &  &  \\
  &  & \ddots &  \\
                 &  &  & \sigma^2 I_r \\
              \end{pmatrix}
\]
Then the variance matrix
\[
V=\sigma^2 I +\sum_{i=1}^r \sigma_i^2 Z_iZ_i^T.
\]
Let $\theta=(\sigma^2, \sigma_1^2, \ldots, \sigma_r^2 )$, then $\ddot{V}_{ij}=0$. For such cases, the close form of inverse of the variance-covariance matrix  is available

\subsection{Averaged information splitting matrix}
In general, when $\ddot{V}\neq 0$, $\I_A$ defined in \eqref{eq:Ia}, is not the average of the observed and expected information but only a main part of it.
The following theorem gives a precise and concise mathematical explanation \cite{ZGL16,Z16}.
\begin{theorem}
Let $\I_O$ and $\I$ be the observed information matrix and the Fisher information matrix for the residual
log-likelihood of the linear mixed model respectively, then the average of the observed and the Fisher information can be split as $\frac{\I_O+\I}{2}=\I_A+ \I_Z$, such that the expectation of $\I_A$ is the Fisher information matrix and $E(\I_{Z})=0$.  \label{thm:S}
\end{theorem}
\begin{proof}
Let the elements of $\I_A$ are defined as in \eqref{eq:Ia} then
apply the result in \eqref{eq:yphphpy}, we have
\begin{equation}
E(\I_A)=\I.
\end{equation}
On the other hand, we have
\begin{equation}
\I_Z(\theta_i,\theta_j)= \frac{\I +I_o}{2} -\I_A =\frac{\tr(P\ddot{V}_{ij})-y^TP\ddot{V}_{ij}Py}{4}.
\end{equation}
Notice that
\[
E(y^TP\ddot{V} Py) = E( \tr(y^TP\ddot{V}Py) ) =E(\tr(Pyy^T P\ddot{V}_{ij})) =\tr(PVP\dot{V}) =\tr(P\ddot{V}_{ij}).
\]
then we have
and $E(\I_Z)=0$.
\qed
\end{proof}
Theorem \ref{thm:S} indicates that approximate information matrix $\I_A$ defined in \eqref{eq:Ia} is the essential main part of the average of the observed information and the Fisher information matrix. To tell the difference with the real average , we can refer $\I_A$ as the averaged information splitting matrix.  It is a good approximation to the Fisher information matrix, but unlike the Fisher information matrix which is data independent, the averaged information splitting matrix is a data dependent.
\begin{algorithm}[!t]
\caption{Average information(AI) algorithm  to solve
{$S(\theta)=0$.}}
\begin{algorithmic}[1]
\State {Give an initial guess of $\theta_0$} \For{ $k=0, 1, 2,
\cdots$ until convergence }
 \State{Solve $\I_A(\theta_k) \delta_k=S(\theta_k)$,}
 \State{$\theta_{k+1}=\theta_k+\delta_k$}
\EndFor
\end{algorithmic}
\label{alg:AI}
\end{algorithm}

\section{Computing issues}

Compare $\I_A $ with $\I_O$, and $\I_F$ in Table \ref{tab:splitting}, in contrast with $\I_{O}(\theta_i,\theta_j)$ which involves 4 matrix-matrix products, $\I_{A}(\theta_i, \theta_j)$ only involves a quadratic term, which can be evaluated by four matrix-vector multiplications and an inner product as in Algorithm~\ref{alg:IKK}.This provide a simple formula. Still the matrix vector multiplication of $Py$ involves the inverse of the $H$ which is of order $n\times n$. When the observations is greater than the number of fixed and random effects, say $n>p+b$, we can obtain a much simpler matrix vector multiplication as $R^{-1}e$, where $e$ is the fitted residual $e= y-X\hat{\tau} -Z\tilde{u}$.

\begin{algorithm}
\caption{Compute $\I_A(\kappa_i,\kappa_j)=\frac{y^TP\dot{V}_iP\dot{V}_jPy}{2\sigma_2}$}
\label{alg:IKK}
\begin{algorithmic}[1]
\State{ $\xi =Py$ }
\State{ $\eta_i =V_i \xi$; $\eta_j=V_j\xi$};
\State{ $\zeta_i= P\eta_j $}
\State{ $\I_A(\theta_i ,\theta_j)=\frac{\eta_i^T \zeta_i}{2\sigma^2}$}
\end{algorithmic}
\end{algorithm}
 We introduce the following lemma.
\begin{lemma}
The inverse of the matrix $C$ in \eqref{eq:mme} is given by
\[
C^{-1} =
\begin{pmatrix}
A& B \\
B^T & D
\end{pmatrix}^{-1}=
\begin{pmatrix}
C^{XX}  & C^{XZ} \\
C^{ZX} & C^{ZZ} \\
\end{pmatrix}
\]
where
\begin{align}
C^{XX} & =(X^TH^{-1}X)^{-1},\\
C^{XZ} & =-C^{XX}X^TR^{-1}ZD^{-1}, \\
C^{ZX} &= -D^{-1} Z^TR^{-1}XC^{XX}, \\
C^{ZZ} & =D^{-1}+C_{ZZ}^{-1}Z^TR^{-1}XC^{XX}X^TR^{-1} Z^TD^{-1}.
\end{align}
\label{lem:C1}
\end{lemma}
\begin{proof}
According to Fact~\cite[Fact 2.17.3]{B09} on $2\times 2$ partitioned matrix, $C^{-1}$ is given by
$$
\begin{pmatrix}
S^{-1}   & -S^{-1}BD^{-1} \\
-D^{-1} B^TS^{-1} & D^{-1}B^TS^{-1}BD^{-1}+D^{-1}.
\end{pmatrix}
$$
where $S=A-BD^{-1}B^T$. So we only need to prove
\begin{align*}
C^{XX} & =((X^TR^{-1}X)^{-1}-(X^TR^{-1}Z) D^{-1}(Z^TR^{-1}X))^{-1} \\
 & = (X^T \underbrace{(R^{-1}-R^{-1}Z(Z^TR^{-1}Z+G^{-1})^{-1}Z^TR^{-1})}_{H^{-1}    \text{see eq.} \eqref{eq:H-1}} X)^{-1}
 =(X^TH^{-1}X)^{-1}.
\end{align*}
\qed
\end{proof}

We shall prove the following results
\begin{theorem}
Let $P$ be defined in \eqref{eq:PHX}, $\hat{\tau}$ and $\tilde{u}$ be the solution to \eqref{eq:mme}, and $e$ be the residual $e=y-X\hat{\tau}-Z\tilde{u}$, then
\begin{align}
P=H^{-1} -H^{-1}X(X^TH^{-1}X)^{-1}X^TH^{-1}
 =R^{-1} -R^{-1}WC^{-1}W^TR^{-1}, \label{eq:P2}
\end{align}
where $W=[X,Z]$ is the design matrix for the fixed and random effects and
\[
 Py=R^{-1} e.
\]
\end{theorem}
\begin{proof} Suppose \eqref{eq:P2} hold, then
\begin{align}
Py =R^{-1}y-R^{-1}W \underbrace{C^{-1}W^TR^{-1} y}_{(\hat{\tau}^T, \tilde{u}^T)^T}
   =R^{-1}(y -X\hat{\tau} -Z\tilde{u}) =R^{-1}e.
\end{align}
\begin{align*}
&  R^{-1} -R^{-1}WC^{-1}W^TR^{-1}
= R^{-1}
-R^{-1}(X, Z) \begin{pmatrix} C^{XX} &  C^{XZ}  \\ C^{ZX} & C^{ZZ}
\end{pmatrix} \begin{pmatrix} X^T \\Z^T \end{pmatrix} R^{-1}. \\
= & R^{-1} -R^{-1}\{ XC^{XX}X^T -XC^{XZ}Z -ZC^{ZX} + ZD^{-1}Z \\
& + Z(C_ZZ^{-1}Z^TR^{-1}XC^{XX}X^TR^{-1} Z^TD^{-1})Z^T \}R^{-1} \\
 = &\underbrace{R^{-1} -R^{-1}ZD^{-1}Z^TR^{-1}}_{H^{-1}}
  -(R^{-1} -R^{-1}ZD^{-1}Z^TR^{-1}) XC^{XX}X^{T}H^{-1} \\
= & H^{-1} -H^{-1}X(X^TH^{-1}X)^{-1}X^TH^{-1}.
\end{align*} \qed
\end{proof}
From above results, we find out that evaluating the matrix vector $Py$ is equivalent the solve the linear system \eqref{eq:mme}
\begin{equation}
\begin{pmatrix}
X^TR^{-1}X  & X^T R^{-1}Z \\
Z^TR^{-1}X  & Z^TR^{-1}Z+G^{-1}
\end{pmatrix}\begin{pmatrix}
\hat{\tau}  \\ \tilde{u}
\end{pmatrix}
=\begin{pmatrix}
X^TR^{-1}y \\ Z^T R^{-1}y
\end{pmatrix}.
\end{equation}
and then evaluate the weighted residual $R^{-1}e$. Notice that the matrix $P \in \mathbb{R}^{n\times n}$. On contrast, $C \in \mathbb{R}^{(p+b)\times (p+b)}$ where $p+b$ is the number of fixed effects and random effects. This number $p+b$ is much smaller than the number of observations $n$. In each nonlinear iteration, the matrix $C$  can be pre-factorized as $C=LDL^T$ by efficient sparse factorization methods like \cite{alg849}\cite{misz93}\cite{masud13}. And then calculate the intermediate variable $\xi =Py=R^{-1}e$. The factorization of $C$ can be reused on one hand to evaluate component of the restricted log-likelihood \eqref{eq:lrc}
\[
\log \lvert C\rvert = \sum_{i=1}^{p+b} \log d_{ii}.
\]
(we assume that other terms in \eqref{eq:lrc} is easer to obtain.)
On the other hand the factorization can be reused in line 6 in Algorithm~\ref{alg:IKK2}.

\begin{algorithm}
\caption{Compute $\I_A(\kappa_i,\kappa_j)=\frac{y^TP\dot{V}_iP\dot{V}_jPy}{2\sigma_2}$}
\label{alg:IKK2}
\begin{algorithmic}[1]
\State{ Give $X$,$Z$,$R$,$G$ and current $\theta$, Let $W=[X, Z]$.}
\State{ Factorize $C=LDL^T$ }
\State{ Solve the mixed model equation $C \beta  = W^TR^{-1}y$; }
\State{ Calculate $e_y =y-W\beta$, $\xi= R^{-1}e_y$};
\State{ Calculate $Y=\{\dot{V}_1\xi, \ldots, \dot{V}_r \xi \}$}
\State{ Solve the mixed model equations with multiple right hand sides $ (LDL^T) B =W^TR^{-1}Y$}
\State{ Calculate $E_Y=Y-WB$; $\Xi=R^{-1}E_Y$}
\State{ $\I_A= Y^T \Xi$/2}.
\end{algorithmic}
\end{algorithm}

\section{Discussion}

The paper details that the elements of an approximate Hessian matrix of the log-likelihood function can be computed by solving the mixed model equations.
Matrix transforms play an important role in splitting the average Jacobian matrices of the score function. Such a splitting results in a simple approximated Jacobian matrix which reduces computations
form four matrix-matrix multiplications to four matrix-vector multiplications. This significantly reduces the time for evaluating the Jacobian matrix in the Newton method.  The problem of evaluating the Jacobian matrix of the score function finally is reduced to solving the mixed model equations \eqref{eq:mme} with multiple right hand sides. At the end of the day, an efficient sparse factorization method plays a crucial role in evaluation the Hessian matrix of the log-likelihood function.

\section*{Reference}
\bibliographystyle{plain}

\end{document}